\documentclass{article}

\usepackage[pdftex]{graphicx}

\usepackage{amssymb,amsfonts,amsmath}
\usepackage{bm}
\usepackage{amsthm}
\usepackage{eucal}

\newtheoremstyle{break}  
  {\topsep}   
  {\topsep}   
  {\itshape}  
  {0pt}       
  {\normalfont\normalsize\bfseries} 
  {.}         
  {5pt plus 1pt minus 1pt}  
  {}          
\theoremstyle{break}

\def\@begintheorem#1#2{\trivlist
   \item[\hskip \labelsep{\theoremfont #1\ #2.}]\itshape}

\def\@opargbegintheorem#1#2#3{\trivlist
      \item[\hskip \labelsep{\theoremfont #1\ #2.\ (#3)}]\itshape}

\begin{document}

\newcommand{\mat}[1]{\ensuremath{\mathbf{#1}}}

\newenvironment{disarray}%
 {\everymath{\displaystyle\everymath{}}\array}%
 {\endarray}
\everymath{\displaystyle\everymath{}}

\newtheorem{theorem}{Theorem}
\newtheorem{lemma}{Lemma}
\newtheorem*{definition}{Definition}
\newtheorem{corollary}{Corollary}
\newtheorem*{conjecture}{Conjecture}

\newtheorem{condition}{Condition}

\def\sgn{\mathrm{sign}} 

\newcommand{\prox}{\operatorname{prox}}    

\def\gradient{\nabla}   

\newcommand{\argmin}[1]{\underset{#1}{\operatorname{argmin}}}

\newcommand{\argmax}[1]{\underset{#1}{\operatorname{argmax}}}

\newcommand{\mypara}[1]{{\noindent \bf{#1.} }}

\newcommand{\bb}[1]{\bm{\mathrm{#1}}}

\def\Tr{\mathrm{T}}
\def\RR{\mathbb{R}}
\def\Sym{\mathrm{Sym}}
\def\Iso{\mathrm{Iso}}
\def\Conv{\mathrm{Conv}}
\def\Aff{\mathrm{Aff}}
\def\trace{\mathrm{tr}}
\def\diag{\mathrm{diag}}

\def\st{\,\,\,\,\mathrm{s.t.}\,\,}

\def\Acal{\mathcal{A}}
\def\Bcal{\mathcal{B}}
\def\Pcal{\mathcal{P}}
\def\Dcal{\mathcal{D}}
\def\Scal{\mathcal{S}}

\def\Ab{\bb{A}}
\def\Bb{\bb{B}}
\def\Cb{\bb{C}}
\def\Db{\bb{D}}
\def\Pb{\bb{P}}
\def\Qb{\bb{Q}}
\def\Db{\bb{D}}
\def\Fb{\bb{F}}
\def\Gb{\bb{G}}
\def\Hb{\bb{H}}
\def\Xb{\bb{X}}
\def\Ib{\bb{I}}
\def\Ub{\bb{U}}
\def\Rb{\bb{R}}
\def\Eb{\bb{E}}
\def\Nb{\bb{N}}
\def\Mb{\bb{M}}
\def\Sb{\bb{S}}

\def\eb{\bb{e}}
\def\fb{\bb{f}}
\def\ub{\bb{u}}
\def\qb{\bb{q}}
\def\rb{\bb{r}}
\def\vb{\bb{v}}
\def\cb{\bb{c}}
\def\deltab{\bb{\delta}}

\newcommand{\ones}{\bb{1}}

\def\Pib{\bb{\Pi}}
\def\Lambdab{\bb{\Lambda}}
\def\alphab{\bb{\alpha}}
\def\gammab{\bb{\gamma}}

\def\Pir{\CMcal{P}}
\def\Dr{\CMcal{D}}
\def\Ir{\CMcal{I}}
\def\CR{\CMcal{C}}
\def\Sropt{\CMcal{S}^\ast}
\def\Dropt{\CMcal{D}^\ast}
\def\Piropt{\CMcal{P}^\ast}

\def\Sr{\mathrm{S}}
\def\Cr{\mathrm{C}}

\def\dis{\mathrm{dis}}

\def\specrad{\sigma}




\title{Graph matching: relax or not?}





\author{Yonathan Aflalo\footnote{Dept. of Computer Science, Technion, Israel} \and 
Alex Bronstein\footnote{School of Electrical Engineering, Tel Aviv University, Israel} \and Ron Kimmel\footnotemark[1] }

\maketitle

\begin{abstract} 
We consider the problem of exact and inexact matching of weighted undirected graphs, in which
a bijective correspondence is sought to minimize a quadratic weight disagreement.  This computationally challenging problem
is often relaxed as a convex quadratic program,
in which the space of permutations is replaced by the space of doubly-stochastic matrices. However, the applicability of such a relaxation
is poorly understood.
We define a broad class of \emph{friendly} graphs characterized by an easily verifiable spectral property. We prove that for friendly graphs, the convex relaxation is guaranteed to find the exact isomorphism or certify its inexistence.
This result is further extended to approximately isomorphic graphs, for which we develop an explicit bound on the amount of weight disagreement under which
the relaxation is guaranteed to find the globally optimal approximate isomorphism.
We also show that in many cases, the graph matching problem can be further harmlessly relaxed to
a convex quadratic program with only $n$ separable linear equality constraints, which is substantially
more efficient than the standard relaxation involving $2n$ equality and $n^2$ inequality constraints.
Finally, we show that our results are still valid for unfriendly graphs if additional information in the form of seeds or attributes is allowed, with the latter satisfying an easy to verify spectral characteristic.

 \end{abstract}





\section{Introduction}

Graphs are a natural abstraction in a variety of problems and are particularly useful for modeling structures,
frequently arising in different domains of science and engineering. In many applications, graphs have to be compared or brought into correspondence.
The term \emph{graph isomorphism} or the less precise term \emph{graph matching} (used mainly in the applied community) refer to a class of computational problems consisting of finding an optimal correspondence between the vertices of two graphs that minimizes adjacency disagreement.
The uses of graph models in general and graph matching in particular are too numerous to allow a comprehensive review within the scope of this paper. In what follows, we will just list a few prominent ones, referring the reader to a (partial) review of applications with a particular emphasis on the domain of pattern recognition \cite{conte2004thirty}.
In computer vision and pattern recognition, graph matching is used for stereo vision and 3D reconstruction \cite{christmas1995structural}, object detection and recognition \cite{lades1993distortion,pelillo1999matching} 
-- in particular optical character recognition \cite{rocha1994shape}, and image and video indexing and retrieval \cite{berretti2001efficient}.
In biometric applications, graph-based techniques have been widely used for identification tasks implemented by means of elastic graph matching.
These include, among other, face recognition and pose estimation \cite{wiskott1997face}, and fingerprint recognition \cite{isenor1986fingerprint}.
In biomedical applications, graphs have been used to model vascular structures and, more recently,
to represent connections between neurons \cite{sporns2005human}.
In data mining, graphs are used to model networks, including the Web and social networks \cite{cook2006mining}.

Despite the tremendous popularity of graph models, graph matching remains a computationally intensive task. In the strict sense, it is computationally intractable as no polynomial algorithms are known for its solution, except for graphs admitting certain particular structures.
The increase in the available computational power of modern computers and the remarkable development of numerous efficient graph matching heuristics have made graph matching feasible for relatively large graphs, counting about a thousand of vertices. However, novel applications such as the analysis of brain graphs -- the so-called \emph{connectomes}, and social networks require matching of graphs with millions if not billions of vertices. These scales are far beyond the reach of existing heuristics. Furthermore, a major disadvantage of graph matching heuristics is that, in general, they are not guaranteed to find the optimal matching, or at least to guarantee how far the found matching is from the optimal one.

\vspace{-3mm}

\paragraph{Contributions.} In this paper, we focus on the family of scalable graph matching heuristics based on continuous (in particular, \emph{convex}) optimization \cite{schellewald2001evaluation}. We analyze the standard convex relaxation of the graph matching problem based on replacing the space of permutations by the space of doubly-stochastic matrices, and make the following contributions:

First, we establish conditions under which the relaxation is equivalent to exact graph matching, in the sense that it is guaranteed to find the exact graph isomorphism if such exists, or certify its inexistence (Theorem~\ref{thm:asym:iso}). The class of graphs on which such an equivalence holds is characterized by an easily verifiable spectral property we call \emph{friendliness}, and is surprisingly large -- practically, as large as the class of asymmetric graphs.

Second, we generalize this result to inexact graph matching, providing an explicit bound on the amount of weight disagreement under which the relaxation is guaranteed to find the globally optimal approximate isomorphism (Theorem~\ref{thm:asym:eps-iso}). 

Third, we show that equivalence of convex relaxation to exact graph matching still holds for unfriendly graphs if additional information besides the graph adjacencies is allowed to disambiguate the symmetries. Specifically, we consider such additional information in the form of a  collection of knowingly corresponding functons (seeds) or vector-valued vertex attributes, and show a constructive spectral condition on the seeds/attributes under which convex relaxation of seeded/attributed graph matching is guaranteed to find one of the isomorphisms.
(Theorem \ref{theorem:symmetric}).

These three contributions establish the boundaries of applicability of the convex relaxation, which have so far been poorly understood. Finally, a byproduct of our analysis is the fact that the former results are also satisfied by a simpler convex relaxation, in which the space of permutations is replaced with an affine space of matrices we call \emph{pseudo-stochastic}. This alternative relaxation leads to a simpler, and potentially more efficiently solvable, optimization problem.

\vspace{-3mm}

\paragraph{Notation.} The following notation will be used in the rest of the paper: vectors and matrices are denoted in bold lower and upper case, respectively, and their elements by lower and upper case italic with appropriate subscript indices. The norm $\| \cdot \|$ will denote the standard $\ell_2$ norm of a vector, and the spectral norm of a matrix (to be distinguished from the Frobenius norm, specified with the subscript $_\mathrm{F}$). 
Throughout the paper, the not so rigorous term \emph{matching} refers to the exact or inexact graph isomorphism problems rather than to the graph-theoretic notion of an independent edge set.

\section{Graph matching}

Let $\Acal = (V,\Ab)$ and $\Bcal = (V,\Bb)$ be two undirected graphs built upon a common vertex set $V$, which for convenience is assumed to be $V=\{1,\dots,n\}$. As $\Acal$ and $\Bcal$ are fully represented by the symmetric $n\times n$ adjacency matrices $\Ab$ and $\Bb$, we will use the two notations interchangeably. We allow the adjacency matrices to contain non-binary edge weights, and henceforth consider this case without explicitly specifying that the graphs are weighted.
Let us denote by $\Pir(n) = \{ \pi : V \rightarrow V \}$ the space of vertex permutations,
which can be equivalently represented by $n \times n$ permutation matrices of the form
$\{ \Pib \in \{0,1 \}^{n\times n}  : \Pib \ones = \Pib^\Tr \ones = \ones \}$, with
$\ones$ denoting a column vector of ones. With some abuse of notation, we will refer to both spaces as $\Pir(n)$, dropping the $n$ whenever possible.
A permutation $\pi$ represents a bijective correspondence between the two graphs, mapping each vertex $i$ in $\Acal$ to a vertex $\pi_i$ in $\Bcal$.
Similarly, for each edge $(i,j)$, the correspondence pulls back the adjacency weight $b_{\pi_i,\pi_j}$.
The latter can be stated equivalently by constructing a new adjacency matrix $\Pib^\Tr \Bb \Pib$,
where $\bb{\Pib}$ is the permutation matrix representing $\pi$.
To measure the adjacency disagreement under correspondence, we define on $\Pir$
a \emph{distortion} function of the form $\dis_{\Acal\mapsto \Bcal}(\Pib) = \| \Ab - \Pib^\Tr \Bb \Pib \|$, with $\| \cdot \|$ denoting some norm (for brevity, we will drop $\Acal\mapsto\Bcal$ whenever possible). 
The graphs are said to be \emph{isomorphic} if there exists a zero-distortion
permutation. We denote the collection of all isomorphisms relating $\Acal$ and $\Bcal$ by
$\Iso(\Acal \mapsto \Bcal) = \{ \Pib : \dis(\Pib) = 0 \}$.

In this notation, the \emph{graph matching} (GM) problem consists of finding a zero-distortion permutation; such a permutation might not be unique if the graph possesses symmetries, as we clarify in the sequel. The closely related \emph{graph isomorphism} (GI) problem consists of verifying whether a zero-distortion permutation exists.
This strict setting 
is usually referred to as \emph{exact}.
Since in practical applications the matched graphs might be contaminated by noise, GM is frequently stated in the \emph{inexact} flavor, consisting of finding a minimum rather than zero distortion permutation.
It is worthwhile noting that the formulation of GM based on a norm of the adjacency disagreement is extremely popular in computer vision, shape analysis \cite{bronstein2006generalized}, and neuroscience \cite{vogelstein2011large}, where graphs are used to represent geometric structures, and the matching distortion can be interpreted as the strength of geometric deformation.
While we focus exclusively on this class of problems, several alternative formulations of GM, particularly those based on edit distance \cite{gao2010survey} and maximum common subgraph \cite{Pelillo04metricsfor,Pelillo_98} have been extensively addressed in the literature.

Computationally, GM is at least as hard as GI, which is an NP problem presently not known either to be solvable in polynomial time, or be NP-complete. In fact, GI is one of the few problems which, if P $\ne$ NP, might reside in an intermediate ``GI-complete'' complexity class \cite{fortin1996graph}. 
Yet, the GI problem is known to be only ``moderately exponential'' \cite{babai1981moderately};
furthermore, polynomial (and even linear) time algorithms exist
for checking the isomorphism of various particular types of graphs, such as planar graphs  \cite{hopcroft1974linear},
graphs with bounded vertex degree \cite{luks1982isomorphism}, and trees \cite{aho1974design}.
However, the constants characterizing the complexity of such algorithms are extremely large; for example, the linear time algorithm for checking the isomorphism
of graphs with vertex degree bounded by $2$ is over $2 \times 10^6$! Moreover, these results are largely inapplicable to inexact or weighted graph matching.
Because of this fact, exact graph matching is not used in practical applications involving even moderately-scaled graphs, except for very particular cases.
Instead, various types of heuristics are usually employed.

The common property of heuristic algorithms is that they often perform well on real problems and scale to large graphs at the expense of having no theoretical guarantee to converge to the true global minimizer of the GM problem. The wealth of literature dedicated to graph matching heuristics counts hundreds of studies published in the past four decades, and we will not attempt to review it within the scope of this paper. Instead, we refer the reader to \cite{conte2004thirty} for a comprehensive review, and focus on the popular class of \emph{continuous optimization relaxations}. In these heuristics, the combinatorial GM problem is replaced by an optimization problem with continuous variables, enabling the use of efficient and scalable continuous optimization algorithms \cite{bertsekas1999nonlinear}.

\section{Relaxation of graph matching}

Adopting this perspective, GM can be formulated as an optimization problem
\begin{equation}
\Pib^\ast = \argmin{\Pib \in \Pir} \, \dis_{\Acal\mapsto \Bcal}(\Pib) = \argmin{\Pib \in \Pir} \| \Ab - \Pib^\Tr \Bb \Pib \|.
\label{eq:match}
\end{equation}
The norm in the objective is typically chosen to be the standard $\ell_1$ norm $\textstyle{\| \Xb \|_1 = \sum_{i,j} |x_{ij}|}$,
the Frobenius ($\ell_2$) norm $\textstyle{\| \Xb \|_\mathrm{F}^2 = \sum_{i,j} x_{ij}^2}$, or the min-max ($\ell_\infty$) norm
$\textstyle{\| \Xb \|_\infty = \max_{i,j} |x_{ij}|}$.
In what follows, we will adopt the Frobenius norm, henceforth defining
$$
\dis_{\Acal\mapsto \Bcal}(\Pib) = \| \Ab - \Pib^\Tr \Bb \Pib \|_\mathrm{F}  = \| \Pib \Ab - \Bb \Pib \|_\mathrm{F},
$$
where the second identity is possible due to
unitarity of permutation matrices. For this particular choice, problem \eqref{eq:match} can be rewritten as
\begin{equation}
\Pib^\ast = \argmin{\Pib \in \Pir} \| \Pib \Ab - \Bb \Pib \|_\mathrm{F}^2 = \argmax{\Pib \in \Pir}\,\trace{( \Bb \Pib \Ab \Pib^\Tr )},
\label{eq:qap}
\end{equation}
known as a \emph{quadratic assignment problem} (QAP).

Both \eqref{eq:match} and \eqref{eq:qap} are NP-hard due to the combinatorial complexity of the constraint $\Pib \in \Pir$.
Relaxation techniques reduce this complexity by replacing the latter constraint with a more tractable continuous set.
Since the practically used norms in \eqref{eq:match} can yield a convex minimization objective,
\emph{convex relaxation} techniques consist of replacing $\Pir$ with a larger convex set, resulting in a tractable convex program.
Various techniques differ mainly in the choice of the norm, the choice of the convex set (i.e., the relaxation), and the particular
numerical algorithm used to solve the resulting convex program \cite{schellewald2001evaluation}.

A  popular choice is to relax $\Pir$ to the space of doubly-stochastic matrices
$\Dr = \{ \Pb \ge \bb{0}  : \Pb \ones = \Pb^\Tr \ones = \ones \}$ constituting the convex hull of permutation matrices in $\mathbb{R}^{n \times n}$.
Combined with the $\ell_1$ or the $\ell_\infty$ norms, such a relaxation leads to a linear program \cite{almohamad1993linear}, while the use of the Frobenius norm results in
a linearly-constrained quadratic program (LCQP or QP for short) \cite{vogelstein2011large}. Both types of optimization problems are solvable using polynomial time algorithms, very efficient in practice \cite{bertsekas1999nonlinear}.

Along with convex relaxations of the GM problem \eqref{eq:match}, there exist numerous techniques for relaxing its QAP formulation \eqref{eq:qap}. Note that after the relaxation the two problems are generally not equivalent! Unlike \eqref{eq:match}, the objective of \eqref{eq:qap} is non-convex and hence even if $\Pir$ is replaced by a convex set, the resulting optimization problem is non-convex. One of the most celebrated relaxations of QAP is the \emph{spectral relaxation} \cite{leordeanu2005spectral}, in which the solution matrix is constrained to constant Frobenius norm, which transforms the relaxed problem to the maximum eigenvector problem. The latter is one of the few non-convex optimization problems for which there exists algorithms with global convergence guarantees. Other non-convex relaxations of the QAP  have been proposed, including restricting the matrix $\Pb$ to the non-negative simplex \cite{BuloPelillo11}, or to the space of doubly-stochastic matrices \cite{vogelstein2011large}. All such relaxations have only local convergence guarantees.

In this paper, we consider the convex QP relaxation of GM,
\begin{equation}
\Pb^\ast = \argmin{\Pb \in \Dr} \| \Pb \Ab - \Bb \Pb \|_\mathrm{F}^2.
\label{eq:gm_relaxed}
\end{equation}
In the sequel, we show that the double-stochasticity constraint can be further harmlessly relaxed.
Since the solution $\Pb^\ast$ of the relaxation is, generally, not a permutation matrix, it has to be projected back onto $\Pir$ \cite{conte2004thirty}.
The orthogonal projection $\Pb^\ast$ onto $\Pir$
has to maximize the standard Euclidean inner product, which can be stated as the optimization problem
\begin{equation}
\hat{\Pib} = P_{\Pir} \Pb^\ast = \argmax{\Pib \in \Pir}\, \langle \Pib, \Pb^\ast\rangle = \argmax{\Pib \in \Pir} \, \trace( \Pib^\Tr \Pb^\ast ).
\label{eq:lap}
\end{equation}
This problem is called a \emph{linear assignment problem} (LAP) and, unlike QAP, is solvable in polynomial time using a family of techniques collectively
known as the Hungarian method \cite{kuhn1955hungarian}. LAP can also be formulated and solved as a linear program, in which the linear objective is minimized over the polytope $\Dr$ instead of $\Pir$. The solution of such a linear program is guaranteed to be in $\Pir$ due to a particular property of the constraints called total unimodularity.

The considered convex relaxation of graph matching can be thus summarized as the following two-step procedure, which we henceforth call
the \emph{relaxed GM} or RGM:
\begin{enumerate}
\item Solve QP \eqref{eq:gm_relaxed}.
\item Project the obtained solution onto the space of permutation matrices by solving the LAP \eqref{eq:lap}.
\end{enumerate}
We will henceforth refer to the permutation matrix $\hat{\Pib}$ obtained from step $2$ above
as the solution of the RGM.

Variants of the described procedure are often used in practice; due to their relatively low computational complexity, they scale to large graphs.
There is a considerable practical evidence that the RGM produces a good approximation to the exact solution of the GM problem, in the sense that $\dis(\hat{\Pib}) \approx \dis(\Pib^\ast)$, and often $\hat{\Pib} \approx \Pib^\ast$.
It is therefore astonishing that no theoretical bounds exist on
$| \dis(\hat{\Pib}) - \dis(\Pib^\ast) |$, and practically nothing is known about $\| \hat{\Pib} - \Pib^\ast \|$!
One of the main goals of this paper is to establish conditions under which RGM is \emph{equivalent} to the exact GM, in the sense that the projection of the solution space of \eqref{eq:gm_relaxed} onto $\Pir$ coincides with $\Iso(\Acal \mapsto \Bcal)$. We also investigate conditions for the converse situation, when
the solution space of the relaxation contains non-zero distortion permutations, making the relaxation unusable.

\section{Exact matching of asymmetric graphs}

We start with the case of exact (i.e., distortion-less) matching of asymmetric graphs.
An undirected graph $\Acal$ with the adjacency matrix $\Ab$ is said to possess a \emph{symmetry} $\Pib \in \Pir$ if $\dis_{\Acal\mapsto \Acal}(\Pib) = 0$. This notation
emphasizes that symmetries are self-isomorphisms. Symmetries form a group with the matrix multiplication operation (or with the function composition, if permutations are interpreted as bijective functions), which we refer to as the \emph{symmetry group} (a.k.a.  \emph{automorphism group}) of $\Acal$ and denote by $\Sym\, \Acal$.
The graph is called \emph{asymmetric} if its symmetry group is trivial, $\Sym\, \Acal = \{ \Ib \}$.

It is straightforward to show that two isomorphic graphs $\Acal$ and $\Bcal$ have identical (isomorphic) symmetry groups, and if $\Pib \in \Pir$ is an isomorphism,
then $\Pib \, \Sym^\Tr\Acal = \{ \Pib \Pib^{\prime \Tr} : \Pib^{\prime} \in \Sym\, \Acal \}$ (or, equivalently, $\Sym\,\Bcal \, \Pib$) are also isomorphisms.
The converse is also true: if the two graphs are related by a collection of isomorphisms $\Iso(\Acal \mapsto \Bcal) = \{ \Pib_1, \dots, \Pib_k  \in \Pir \}$,
then they are symmetric with $\Sym\,\Acal$ generated by $\Iso(\Bcal \mapsto \Acal) \circ \Iso(\Acal \mapsto \Bcal) = \{ \Pib_i^\Tr \Pib_j \}$, and
$\Sym\,\Bcal$ by $\Iso(\Acal\mapsto \Bcal) \circ \Iso(\Bcal\mapsto \Acal) = \{ \Pib_i \Pib_j^\Tr \}$.
Consequently, if $\Acal$ is asymmetric and $\Bcal$ is isomorphic to it, they are related by a unique isomorphism which is the global minimizer of \eqref{eq:match}.
In what follows, we denote this unique isomorphism by $\Pib^\ast$.

We emphasize that the symmetry or asymmetry of a graph has nothing to do with the fact that the adjacency matrix $\Ab$ is symmetric.
The latter property stems from the fact that the graph is undirected, and because of it $\Ab$ admits unitary diagonalization of the form $\Ab = \Ub \Lambdab \Ub^\Tr$,
with an orthonormal $\Ub = (\ub_1,\dots, \ub_n)$ containing the eigenvectors in its columns, and a diagonal $\Lambdab = \diag\{\lambda_1,\dots,\lambda_n \}$ containing
the corresponding eigenvalues.

The uniqueness of the isomorphism relating isomorphic asymmetric graphs is crucial for the results we present next.
However, the existence or the absence of symmetry is not an easy property to verify.
To overcome this difficulty, instead of considering the class of asymmetric graphs, we consider another class of graphs
characterized by the following spectral property:

\begin{definition}
A graph $\Acal$ is called \emph{friendly} if its adjacency matrix $\Ab$ has simple spectrum (i.e., all the $\lambda_i$ are distinct),
and all its eigenvectors satisfy $\ub_i^\Tr \bb{1} \ne 0$.
\label{def:friendly}
\end{definition}

\noindent We note the following important consequence of friendliness:
\begin{lemma}
A friendly graph is asymmetric.
\label{lemma:asymmetry}
\end{lemma}
\begin{proof}
Let $\Ab = \Ub\Lambdab\Ub^\Tr$ denote the adjacency matrix of the graph,
and let assume by contradiction that there exists $\Pib \ne \Ib$ such that $\Pib \Ab = \Ab \Pib$.
Then, for every eigenvector $\ub_i$ of $\Ab$, we have $\Ab \Pib \ub_i = \Pib \Ab \ub_i = \lambda_i \Pib \ub_i$, that is, $\Pib \ub_i$ is also an eigenvector
of $\Ab$. Since, due to friendliness, $\Ab$ has simple spectrum, the only two possibilities are $\Pib \ub_i = \pm \ub_i$.
Since we assumed $\Pib \ne \Ib$, there must be at least one eigenvector $\ub_i$ for which $\Pib \ub_i = - \ub_i$. Then, $\ones^\Tr \Pib \ub_i = -\ones^\Tr \ub_i$.
On the other hand, since $\Pib$ is a permutation matrix, $\ones^\Tr \Pib \ub_i = \ones^\Tr \ub_i$. Hence, $\ub_i^\Tr \ones = \bb{0}$ in contradiction to friendliness of $\Ab$.
\end{proof}

\noindent
The converse is not true, as there might exist an asymmetric graph with an unfriendly adjacency matrix. For example, any regular unweighted graph has a constant eigenvector and is thus highly unfriendly; on the other hand, there exist asymmetric regular graphs such as the Frucht graph with $n=12$. 
However, unfriendliness still seems to be a singular property,
and intuition suggests that unfriendly graphs should have measure zero among random asymmetric weighted graphs, and
the class of friendly graphs should be
almost as big as that of asymmetric graphs. We do not pursue the rigorous proof
of this claim, since it might delicately depend on what is meant by ``random''.
We only emphasize that, in contrast to asymmetry, friendliness is an easily verifiable property.

Using the notion of friendliness, we state our first result:
\begin{theorem}
Let $\Acal$ and $\Bcal$ be friendly isomorphic graphs. Then, GM and RGM are equivalent.
\label{thm:asym:iso}
\end{theorem}
\begin{proof}
We consider the relaxation \eqref{eq:gm_relaxed} of GM, denoting by $\Pib^\ast$ the global minimizer of the latter. The minimizer is unique
due to asymmetry.
For any doubly-stochastic matrix $\Pb$,
\begin{equation*}
\Pb\Ab-\Bb\Pb  = (\Pb{\Pib^{\ast\Tr}}\Bb-\Bb\Pb{\Pib^{\ast\Tr}})\Pib^\ast =  \Qb\Bb-\Bb\Qb,
\end{equation*}
where $\Qb=\Pb {\Pib^{\ast\Tr}}$. We therefore reformulate \eqref{eq:gm_relaxed} in terms of $\Qb$ as the minimization of
$f(\Qb)=\frac{1}{2}\|\Qb\Bb-\Bb\Qb\|_\mathrm{F}^2$ w.r.t $\Qb \in \Dr$.
Since the objective $f(\Qb)$ is convex in $\Qb$, and so is the set of double stochastic matrices $\Dr$, the problem  has a global minimum at $\Qb = \Ib$.
It remains to prove  that the minimum is unique.
Since $\Bb$ is symmetric, simple calculus yields the gradient of $f(\Qb)$,
$\nabla_{\Qb} f =\Qb\Bb^2+\Bb^2\Qb-2\Bb\Qb\Bb$.
By omitting the nonnegativity and unit column sum constraints, we further relax the constraint $\Qb\in \Dr$ to $\Qb\ones=\ones$, referring to such matrices as \emph{pseudo-stochastic}.
The Lagrangian of $f$ with the pseudo-stochasticity constraint on $\Qb$ is given by
$L(\Qb,\alphab) = f(\Qb) + \alphab^\Tr(\Qb\ones-\ones) = f(\Qb) + \trace\left((\Qb\ones-\ones)\alphab^\Tr\right)$,
with $\alphab$ denoting the vector of Lagrange multipliers.
Problem \eqref{eq:gm_relaxed} reaches a minimum when
\begin{equation*}
\nabla_{\Qb} L(\Qb,\alphab) = \Qb\Bb^2+\Bb^2\Qb-2\Bb\Qb\Bb+\alphab\ones^\Tr  = 0.
\end{equation*}
Substituting the unitary eigendecomposition $\Bb = \Ub \Lambdab \Ub^\Tr$, the latter equation can be rewritten as
\begin{equation}
\Fb\Lambdab^2+\Lambdab^2\Fb-2\Lambdab\Fb\Lambdab+\gammab \vb^\Tr=0,
\label{eq:F}
\end{equation}
where $\gammab=\Ub^\Tr\alphab$, $\vb = \Ub^\Tr \ones$, and $\Fb=\Ub^\Tr\Qb\Ub$.
It is easy to see that \eqref{eq:F} can be expressed coordinate-wise as
\begin{equation}
\label{eqn:eqn_F_}
F_{ij}(\lambda_i-\lambda_j)^2+v_j\gamma_i=0.
\end{equation}
For every $i=j$, we have $v_i \gamma_i=0$; since the friendliness assumption implies $v_i\neq 0$ for all $i$, we have $\gammab=\bb{0}$.
This yields
\begin{equation}
\label{eqn:eqn_F}
F_{ij}(\lambda_i-\lambda_j)^2=0 \,\,\,\mathrm{for\, every}\, i\neq j.
\end{equation}
Since friendliness also implies $\lambda_i \ne \lambda_j$, $\Fb$ must be diagonal.
Because $\Qb$ is pseudo-stochastic, it has to satisfy $\ones=\Qb\ones = \Ub\Fb\Ub^\Tr\ones$ or, equivalently, $\vb = \Fb \vb$.
Yet, since $\Fb$ is diagonal and, due to friendliness, $\vb$ has no zero elements, the above property is satisfied only if
$\Fb=\Ib$. This implies that $\Qb=\Ib$ or, equivalently, $\Pb=\Pib^\ast$. Hence, $\Pib^\ast$ is the unique minimizer of
\eqref{eq:gm_relaxed}. Since the solution is already in $\Pir$, the projection \eqref{eq:lap} leaves it unchanged.
%
\end{proof}

Note that in the proof we only used the pseudo-stochasticity constraint $\Pb \bb{1} = \bb{1}$.
This leads to an important consequence: instead of relaxing $\Pir$ to the space $\Dr$ of doubly-stochastic matrices, a coarser relaxation to
pseudo-stochastic matrices is equivalent in the discussed case. Practically, this means that we can solve a simpler quadratic program
\begin{equation}
\Pb^\ast = \argmin{\Pb} \| \Pb \Ab - \Bb \Pb \|_\mathrm{F}^2 \st \bb{P}\bb{1} = \bb{1},
\label{eq:gm_relaxed1}
\end{equation}
with $n^2$ variables and only $n$ equality constraints, instead of $2n$ equality constraints and $n^2$ inequality constraints in \eqref{eq:gm_relaxed}.
In what follows, we focus on this simpler convex relaxation instead of \eqref{eq:gm_relaxed} in the RGM.

While checking the friendliness condition in Theorem~\ref{thm:asym:iso} is straightforward,
checking whether the perfect isomorphism condition is satisfied is not (in fact, it is a graph isomorphism problem!)
However, in practice one can simply solve relaxation \eqref{eq:gm_relaxed1} for the two friendly graphs, project the result onto $\Pir$, and check whether
$\dis (\hat{\Pib}) = 0$. If the answer is positive, $\hat{\Pib}$ is guaranteed to be the unique isomorphism; otherwise, the theorem guarantees that the graphs are not isomorphic.



\section{Inexact matching of asymmetric graphs}

The case of perfectly isomorphic graphs, to which Theorem~\ref{thm:asym:iso} is applicable, is often an unachievable mathematical idealization.
Many practical applications of graph matching assume some amount of noise, and seek a least distortion correspondence rather than a perfect isomorphism. To formalize this notion,
we say that two graphs
$\Acal$ and $\Bcal$ are $\rho$-isomorphic if there exists $\Pib^\ast \in \Pir$ with $\dis(\Pib^\ast) \le \rho$.

%
Similarly, we say that a graph $\Acal$ possesses an \emph{$\rho$-symmetry} $\Pib \in \Pir$ if $\dis_{\Acal\mapsto \Acal}(\Pib) \le \rho$.
Note that unlike their exact counterparts, $\rho$-symmetries do not form a group, as the composition of two $\rho$-symmetries might have $\dis > \rho$.
A graph with a trivial $\rho$-symmetry set is called \emph{$\rho$-asymmetric}.
The lack of symmetry of such a graph is strong enough to guarantee that a bounded perturbation of the adjacency weights does not produce new,
previously inexistent symmetries.

In order to generalize our result to the case of nearly-isomorphic graphs, we define the strength of a graph's friendliness:
\begin{definition}
A graph $\Acal$ is $(\epsilon,\delta)$-\emph{friendly} if its adjacency matrix $\Ab = \Ub \Lambdab \Ub^\Tr$ has the
spectral gap $\displaystyle{\sigma(\Ab) = \min_{i\ne j} |\lambda_i - \lambda_j| > \delta}$, and
$\displaystyle{\epsilon < |\ub_i^\Tr \bb{1}| < \frac{1}{\epsilon}}$ for $i=1,\dots,n$.
\label{def:eps-friendly}
\end{definition}
\noindent
Also note that our former definition of friendliness corresponds to $(\epsilon,\delta) = (0,0)$.
We refer to the case $\epsilon,\delta > 0$ as \emph{strong} friendliness.

For the broad family of strongly friendly graphs, we first show that the result of Theorem~\ref{thm:asym:iso} is stable in the sense that a bounded perturbation
in the adjacency matrix results in a bounded perturbation of the solution:
\begin{lemma}
Let $\Acal$ and $\Bcal$ be $(\epsilon,\delta)$-friendly isomorphic graphs with spectral radius $\displaystyle{\specrad = \max_{i} |\lambda_i|}$,
related by the unique isomorphism $\Pib^\ast$.
Let $\tilde{\Bcal}$ be a perturbed version of $\Bcal$ with $\tilde{\Bb} = \Bb + \rho \Rb$, where $\Rb$ is symmetric with $\|\Rb\|_\mathrm{F} \le 1$, and $\displaystyle{ \rho  \le \min\left\{ \sqrt{2}\specrad,
\frac{\delta^2 \epsilon^4}{12  \specrad n^{1.5}}
\right\} }$.
Then, the solution $\Pb^\ast$ of the perturbed problem \eqref{eq:gm_relaxed1} is unique and satisfies $\| \Pb_\rho^\ast - \Pib^\ast \|_\mathrm{F} < \textstyle{\frac{1}{2}}$.
\label{lemma:stability}
\end{lemma}

\noindent The proof closely follows the proof of Theorem~\ref{thm:asym:iso}, and relies on
a result from perturbation analysis of linear systems. Full proof as well as the mentioned result (summarized as Lemma~\ref{lemma:perturbation}) are presented in the Appendix.

Applying the former result to matching a graph with itself ($\Acal = \Bcal$), the following generalization of Lemma~\ref{lemma:asymmetry} can be straightforwardly obtained:
\begin{corollary}
\label{cor:strong-asymm}
An $(\epsilon,\delta)$-friendly graph is $\rho$-asymmetric, with $\rho$ satisfying the conditions of Lemma~\ref{lemma:stability}.
\end{corollary}
\noindent In fact, this property guarantees that the perturbation creates no symmetries and, thus,
the perturbed version of system \eqref{eqn:eqn_F} remains full rank.

The stability of the relaxation in Lemma~\ref{lemma:stability} leads directly to our second result:
\begin{theorem}
Let $\Acal$ be an $(\epsilon,\delta)$-friendly graph with the adjacency matrix
normalized such that $\specrad = 1$, and let $\Bcal$ be 
$\rho$-isomorphic to $\Acal$.
Then, if $\rho < \frac{\delta^2 \epsilon^4}{12 n^{1.5}}$, RGM and GM are equivalent.
\label{thm:asym:eps-iso}
\end{theorem}
\begin{proof}
Let $\Pib^\ast$ be a $\rho$-isomorphism relating $\Bcal$ and $\Acal$,
and let us denote $\Bb_0 = \Pib^{\ast} \Ab\Pib^{\ast \Tr}$ and $\Rb = \frac{1}{\rho}(\Bb-\Bb_0)$.
Then, $\Bcal_0$ is perfectly isomorphic to $\Acal$, and
$\Bcal$ is a perturbed version of $\Bcal_0$ with $\Bb = \Bb_0 + \rho \Rb$ and
$\| \Rb \|_\mathrm{F} = \frac{1}{\rho} \| \Bb - \Pib^{\ast } \Ab\Pib^{\ast \Tr} \|_\mathrm{F} \le 1$.
By Corollary~\ref{cor:strong-asymm}, $\Bcal$ is $\rho$-asymmetric and, hence, $\Bcal_0$ is asymmetric.
Denoting by $\Pb^\ast$ the solution of \eqref{eq:gm_relaxed1} applied to $\Acal$ and $\Bcal$,
we invoke Lemma~\ref{lemma:stability} which guarantees uniqueness of $\Pb^\ast$ and $\| \Pb^\ast - \Pib^\ast \|_\mathrm{F} < \textstyle{\frac{1}{2}}$.
By standard norm inequalities, this implies $|P^\ast_{ij} - \Pi^\ast_{ij}| < \textstyle{\frac{1}{2}}$  element-wise for every $i,j$.
Therefore, the projection of $\Pb^\ast$ onto $\Pir$ coincides with $\Pib^\ast$.
%
\end{proof}

As in the case of perfectly isomorphic graphs, checking the strong friendliness condition in Theorem~\ref{thm:asym:iso} is straightforward, while
checking the $\rho$-isomorphism of $\Acal$ and $\Bcal$ is not. Yet, as in the previous case, one can again
solve relaxation \eqref{eq:gm_relaxed1}, project the solution onto $\Pir$, and verify whether $\dis(\hat{\Pib}) < \rho$.
In case of a positive answer, $\hat{\Pib}$ is guaranteed to be the unique global minimizer of the graph matching problem;
otherwise, the graphs are guaranteed not to be $\rho$-isomorphic.
An empirical evaluation of the bound in Theorem~\ref{thm:asym:eps-iso} is presented in Figure \ref{fig:bound}.

\begin{figure}[tb]
\begin{center}
\includegraphics[width=0.75\linewidth]{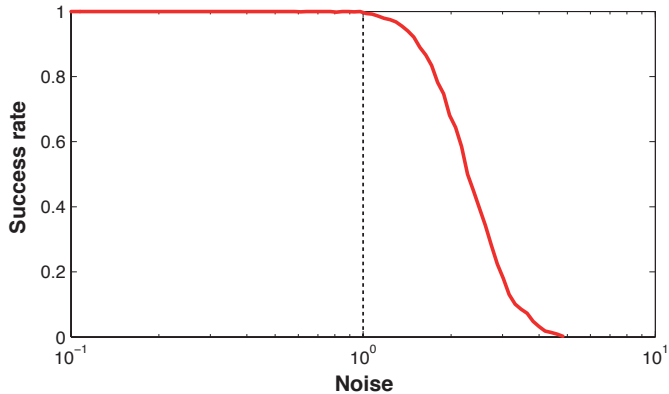}
\end{center}
\caption{Empirical evaluation of the bound in Theorem~\ref{thm:asym:eps-iso} on $10^3$ random strongly friendly graphs. For each graph, different amount of noise was added, and the ratio of successful runs of convex relaxation \eqref{eq:gm_relaxed1} recorded on the vertical axis (a run was deemed successful if the ground truth isomorphism is recovered). The noise strength on the horizontal axis is normalized for each graph in such a way that the value of the bound is always $1$. Observe that all runs with noise within the bound converged successfully, while those with stronger noise failed with probability increasing as the amount of noise grows. \label{fig:bound}} 
\end{figure}

\section{Matching of symmetric graphs}

\begin{figure}[tb]
\begin{center}
\includegraphics[width=0.75\linewidth]{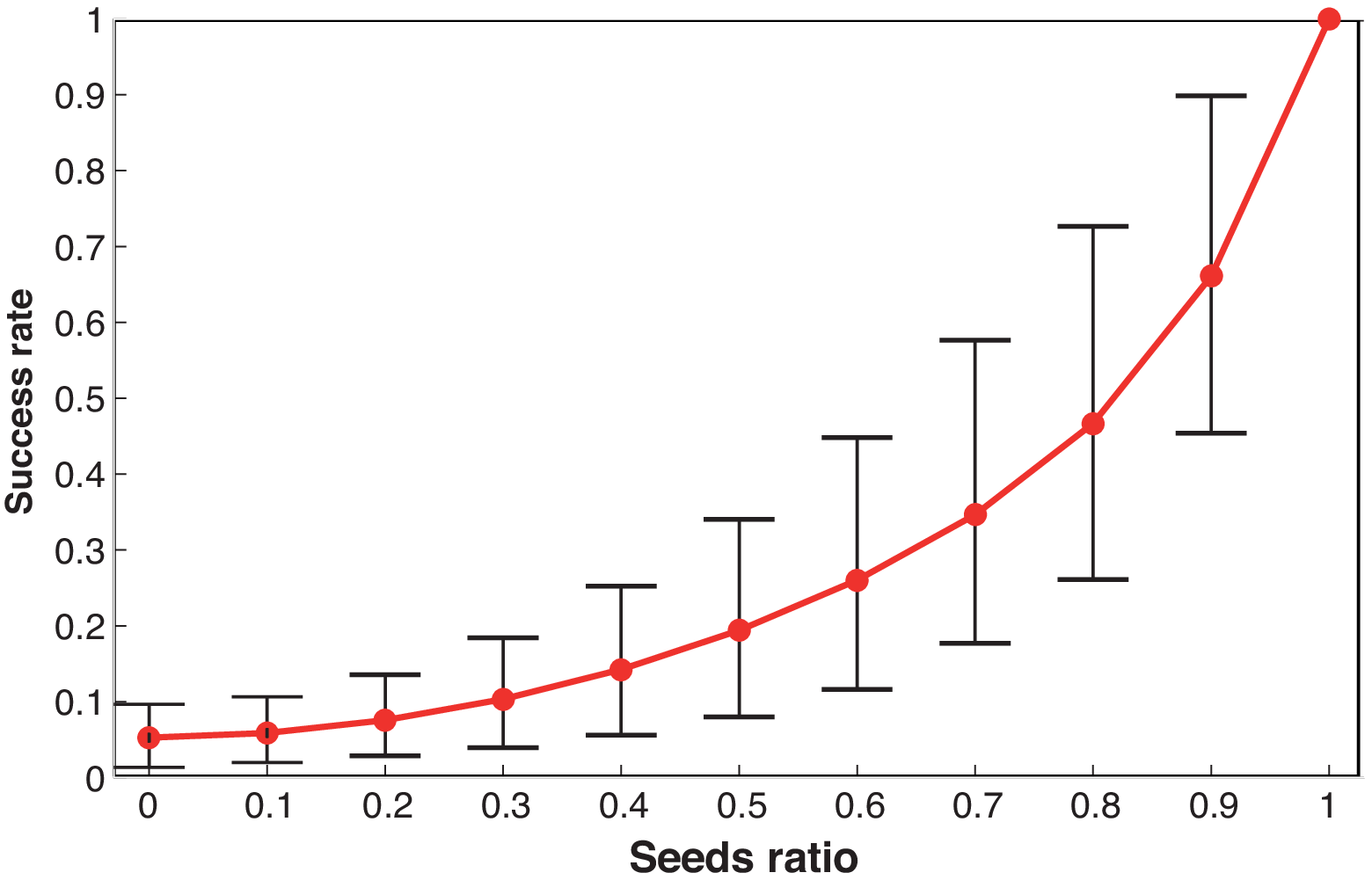}
\end{center}
\caption{Empirical evaluation of convex relaxation of seeded unfriendly graph matching on multiple graphs of different sizes and with different number and types of symmetries. Seeded matching was performed with a different number of random point seeds (plotted on the horizontal axis as the ratio with the number of non-trivial symmetries), all of which were made not invariant under the corresponding number of non-trivial symmetries. The vertical axis represents the success rate of recovering one of the exact isomorphisms. Average, minimum and maximum values are plotted as the red line and the error bars, respectively. 
Two exteme points on the horizontal axis are remarkable in particular: The leftmost point corresponding to unseeded graph matching, showing an empirical evidence to the fact that convex relaxation fails on unfriendly graphs. The failure rate depends on the number of non-trivial symmetries. The rightmost point corresponds to seeded graph matching with the seeds fully disambiguating the symmetries. In this case, perfect recovery of one of the isomorphisms is achieved. 
\label{fig:failure}} 
\end{figure}

The assumption of friendliness plays a crucial role in the results we have developed so far: it guarantees uniqueness of solution of the relaxation. 
These results cannot be directly extended to symmetric graphs, for which the solution space of the relaxation should contain all isomorphisms and their affine combinations. Unfortunately, besides the true isomorphisms, this affine subspace may also contain pseudo-stochastic matrices that are not permutations, some of which falling into Voronoi cells of permutations that are not isomorphisms. As the result, using convex relaxation for matching symmetric graphs may lead to a wrong solution, depending on the particular optimization algorithm and its initialization. 
An empirical evidence of this phenomenon is presented in Figure \ref{fig:failure}; see \cite{lyzinski2014relax} for a formal proof of failure of convex relaxation on a particular class of random graphs.
Yet, in what follows we show that by providing additional information in the form of corresponding seeds or vertex attributes disambiguating the symmetry, equivalence of the relaxation to the exact GM problem still holds.

Let $\Cb$ and $\Db$ be $n \times q$ matrices, whose columns are real-valued functions on the vertices of the graphs $\Acal$ and $\Bcal$, respectively.  For example, an indicator function of the $k$-th vertex in the graph is the $k$-th vector of the standard Euclidean basis in $\RR^n$. The matrices $\Cb$ and $\Db$ can be alternatively interpreted as $q$-dimensional vector-valued vertex attributes, with the $k$-th row of $\Cb$ representing the attribute of vertex $k$ in $\Acal$. 
We say that the matrices $\Cb$ and $\Db$ are covariant under a permutation $\Pib$ relating between the graphs if $\Pib \Cb = \Db$. 
With this additional information, we consider the following extension of \eqref{eq:gm_relaxed1}:
\begin{equation}
\Pb^\ast = \argmin{\Pb} \| \Pb \Ab - \Bb \Pb \|_\mathrm{F}^2 + \mu \| \Pb \Cb - \Db  \|_\mathrm{F}^2 \st \bb{P}\bb{1} = \bb{1}.
\label{eq:gm_relaxed_seeded}
\end{equation}
This problem can be thought of as a convex relaxation of \emph{seeded} graph matching, in which the seeds are  provided through a penalty, whose strength is controlled by the parameter $\mu$, rather than through a hard constraint; alternatively, it can be interpreted as a relaxation of \emph{attributed} graph matching, 
in which a permutation is sought to minimize the aggregate of edge adjacency and vertex attribute disagreement. In light of this duality, we henceforth refer to $\Cb$ and $\Db$ as to seeds.

As before, in order to avoid verifying whether a graph is symmetric or not, we consider the easily verifiable friendliness property. 
We assume that a general adjacency matrix of a graph has $d$ non-simple eigenspaces with multiplicities summing up to $m+d$.  To simplify notation, we will say that an eigenvalue $\lambda_i$ has multiplicity $m_i$, referring to the multiplicity of the eigenspace to which $\lambda_i$ belongs.
%
%
Since the eigenvectors spanning an $m_i$-dimensional eigenspace are defined up to a rotation within it, such eigenvectors shall be selected that none of them is orthogonal to the constant vector $\ones$, unless the entire eigenspace is orthogonal to it. We call the latter eigenspaces \emph{hostile}, and denote by $k$ the total dimension of hostile eigenspaces. A graph is friendly if and only if both $m$ and $k$ are $0$, and is $(m,k)$-\emph{unfriendly} otherwise.

The relation between the size of the symmetry group of a graph and its degree of unfriendliness is summarized in the following result:
\begin{lemma}
\label{lemma:symmetry}
Let $\Acal$ be a graph with $l=|\Sym\,\Acal|-1$ non-trivial symmetries.
Then, $\Acal$ is $(k,m)$-unfriendly with $k+m \ge l$.
\label{lemma:symmetry}
\end{lemma}
\begin{proof}
The proof extends the proof of Lemma~\ref{lemma:asymmetry}, where we showed that for every non-trivial symmetry $\Pib$,
there exists at least one $i$ such that $\Pib \ub_i \ne \ub_i$ is an eigenvector of $\Ab$ corresponding to $\lambda_i$.
Furthermore, if $\ub_i$ is simple, $\ub_i^\Tr \ones = 0$. Therefore, each non-trivial symmetry increments by one either $k$ (in case $\lambda_i$ is simple) or $m$ (otherwise), or both.
\end{proof}

It is easy to observe that with $\mu=0$, each $m_i$-dimensional non-simple eigenspace of the adjacency matrix decreases the rank of system \eqref{eq:F} by $m_i$; if the eigenspace is hostile, the rank is further decreased by one. This is precisely the reason for Theorem \ref{thm:asym:iso} not being applicable to unfriendly graphs.
The introduction of the seeds disagreement term to the objective contributes to \eqref{eq:F} a term of the form $\mu \Fb \Gb$, where $\Gb = \Ub^\Tr \Db \Db^\Tr \Ub$ is the Gram matrix of the seeds $\Db$ represented in the eigenbasis of the adjacency matrix. Since $\Gb$ is always positive semi-definite, the rank of system \eqref{eq:F} typically increases and becomes full under the following conditions:

\begin{theorem}
\label{theorem:symmetric}
Let $\Acal$ and $\Bcal$ be isomorphic unfriendly graphs with adjacencies $\Ab$ and $\Bb$ and seeds $\bb{C}$ and $\bb{D}$, respectively. Let $\bb{C}$ and $\bb{D}$ be covariant under a particular isomorphism $\Pib^\ast \in \Iso(\Acal\mapsto \Bcal)$, and let $\bb{D}$ further satisfy for every non-simple $m_i$-dimensional eigenspace of  $\Bb$ corresponding to $\lambda_i = \cdots = \lambda_{i+m_i}$, $(\ones^\Tr \ub_i) \Db \Db^\Tr \ub_j \ne \ones (\ub_i^\Tr \Db\Db^\Tr \ub_j)$ for every $j = i+1,\dots, i+m_i$ if the eigenspace is not hostile, or $\Db \Db^\Tr \ub_j \ne \bb{0}$ for every  $j = i,\dots, i+m_i$ otherwise.
Then, $\Pib^\ast$ is the the unique minimizer of \eqref{eq:gm_relaxed_seeded} for any $\mu > 0$.
\label{theorem:symmetric}
\end{theorem}

\noindent For a proof, see Appendix.
Conditions of Theorem \ref{theorem:symmetric} are both easy to verify and are \emph{constructive} in the sense that given the spectral decomposition of the adjacency matrix of one of the graphs, the Theorem specifies how to construct a set of seeds such that if a set of corresponding seeds in the other graph is further given and is covariant under a preferred isomorphism, the convex relaxation \eqref{eq:gm_relaxed_seeded} is guaranteed to find the latter isomorphism. In particular, $\Db$ must have at least $m+k$ linearly independent columns, which by Lemma \ref{lemma:symmetry} implies that in order to disambiguate $l$ non-trivial symmetries, the number of seeds has to be at least $l$.

In practice, we observed that it is sufficient to generate random point seeds ensuring that the matrix $\Db$ is not invariant under any non-trivial symmetry of the graph, namely $\Pib \Db \ne \Db$ for every $\Pib \in \Sym\, \Bcal \setminus \{ \Ib \}$. An empirical corroboration of this result is presented in Figure \ref{fig:failure} in Supporting Information.

\section{Discussion and conclusion}

In this paper, we considered convex relaxation of the NP graph matching problem. We proposed an easy-to-verify friendliness property, and proved that for friendly graphs, convex relaxation is equivalent to the computationally intractable exact matching; the result extends to inexact matching of strongly friendly graphs. In such cases, convex relaxation is guaranteed to find the exact (or approximate) isomorphism or guarantee its inexistence.
We also showed that convex relaxation is applicable to exact matching of unfriendly graphs (in particular, those possessing non-trivial symmetries), provided that additional information is supplied in the form of seeds or vertex attributes. We showed constructive spectral characteristics that such seeds/attributes have to satisfy in order for the convex relaxation of seeded graph matching to be guaranteed to find one of the isomorphisms.
%

The analysis we presented is inspired in part by \cite{aflalo2013spectral} where matching surfaces    
 is treated as matching metric spaces performed in their spectral domain. 
However, despite this superficial resemblance, our proofs here are based on the spectral properties of the adjacency matrices, in contrast to the those of the graph Laplacian frequently studied in spectral graph theory.

A surprising observation is that none of our results is influenced by the nonnegativity constraints $\Pb \ge \bb{0}$.
While the space of doubly-stochastic matrices is the smallest convex set containing the space $\Pir$ of $n\times n$ permutations, and is therefore the most natural convex relaxation of the latter, our findings question the utility of the non-negativity constraints in graph matching problems, and suggest relaxing $\Pir$
as the bigger affine space $\{ \Pb : \Pb \ones = \Pb^\Tr \ones = \ones \}$.
Also, for the class of friendly graphs on which we were able to prove global convergence of convex relaxations, the column-wise equality constraints $\bb{P}^\Tr \bb{1}=\bb{1}$ has no utility and can be removed. The question whether these constraints are at all needed, and whether they can help extending the applicability of convex relaxation requires further investigation.
From the practical perspective, the removal of the nonnegativity constraints might allow the use of simpler and better scalable convex optimization algorithms.
Furthermore, the removal of the constraints $\bb{P}^\Tr \bb{1}=\bb{1}$ splits the remaining constraints $\bb{P}\bb{1} = \bb{1}$ into $n$ constraints separable with respect to the rows of $\bb{P}$. This allows to employ block-coordinate update schemes operating each time on $n$ variables from one row of $\bb{P}$ only, thus potentially improving the algorithm scalability to large graphs.


\section*{Acknowledgement}
YA and RK are supported by the ERC advanced grant 267414 (NORDIA).
AB is supported by the ERC starting grant 335491 (RAPID). Valuable feedback from Marcelo Fiori, Guillermo Sapiro, and Michael Elad is acknowledged.

\appendix
\section{Appendix}

\begin{lemma}
\label{lemma:perturbation}
Let $\ub_0$ be the solution of a full-rank linear system $\Mb\ub = \cb$,
and let $\ub$ be the solution of the perturbed full-rank system $(\Mb + \rho \Nb) \ub  = \cb$, $\rho > 0$. Then,
\begin{equation}
\label{eq:perturbed_delta}
\| \ub  - \ub_0  \| \leq \frac{\rho \left\|\Mb^{-1}\right\| \|\Nb\| \|\ub_0\|}{1-\rho \left\|\Mb^{-1}\right\| \|\Nb\|  }.
\end{equation}
\begin{proof}
Denoting $\deltab = \ub  - \ub_0$, we have $\Mb \ub_0 = \cb$ and $(\Mb+\rho\Nb)(\ub_0+\deltab)=\cb$,
from where $(\Mb+\rho\Nb)\deltab=-\rho\Nb \ub_0$. The latter is equivalent to
$\deltab=-\rho(\Mb+\rho\Nb)^{-1}\Nb \ub_0$ assuming an invertible $(\Mb+\rho\Nb)^{-1}$.
From the identity
\begin{eqnarray*}
(\Mb+\rho\Nb)^{-1} &=& \Mb^{-1}(\Ib+\rho\Mb^{-1}\Nb)^{-1}  \\
&=& \Mb^{-1}\left(\sum_{i=0}^\infty (-\rho)^i(\Mb^{-1}\Nb)^i\right)
\end{eqnarray*}
and the inequality $\|(\Mb^{-1}\Nb)^i \vb \| \le \|\Mb^{-1}\| ^i\left\|\Nb\right\| ^i\|\vb\|$ holding for every $i \ge 0$ and every $\vb$,
we have
$$
\| \deltab \| = \left\|(\Mb+\rho\Nb)^{-1}\Nb \ub_0\right\| \leq\frac{\left\|\Mb^{-1}\right\| \|\Nb\| \|\ub_0\|}{1-\rho\|\Mb^{-1}\| \left\|\Nb\right\| }.
$$
\end{proof}
\end{lemma}

\begin{proof}[Proof of Lemma~\ref{lemma:stability}]
%
%
%
The proof goes along the lines of the proof of Theorem~\ref{thm:asym:iso}. As before, we reparametrize the optimization in terms
of $\Qb=\Pb {\Pib^{\ast\Tr}}$ instead of $\Pb$.
Denoting by $\Bb = \Ub \Lambdab \Ub^\Tr$ the orthonormal eigendecomposition of $\Bb$ and $\Eb=\Ub^\Tr\Rb\Ub$, we substitute $\tilde{\Bb} = \Ub (\Lambdab + \rho \Eb)\Ub^\Tr$ into a perturbed version of the Lagrangian,
\begin{equation}
\nabla_{\Qb} L(\Qb,\alphab) = \Qb\tilde{\Bb}^2+\Bb^2\Qb-2\Bb\Qb\tilde{\Bb}+\alphab\ones^\Tr  = 0,
\end{equation}
and obtain a perturbed version of \eqref{eq:F},
\begin{eqnarray*}
\lefteqn{(\Fb\Lambdab^2+\Lambdab^2\Fb-2\Lambdab\Fb\Lambdab)  +\rho(\Fb\Eb\Lambdab+\Fb\Lambdab\Eb-2\Lambdab\Fb\Eb) }\\
&& + \gammab \vb^\Tr+\rho^2\Fb\Gb=0 \hspace{3.5cm},
\end{eqnarray*}
where $\Gb = \Eb^2$, and $\Fb$, $\vb$, and $\gammab$ are defined as before. The system can be rewritten coordinate-wise as
\begin{equation}
\label{eqn:eqn_F_per}
F_{ij}(\lambda_i-\lambda_j)^2+v_j\gamma_i+\rho\sum_kF_{ik}\left(E_{kj}(\lambda_j+\lambda_k-2\lambda_i)+\rho G_{kj}\right)=0.
\end{equation}
Substituting $i=j$ and re-arranging the terms yields
\begin{equation*}
\gamma_i=-\frac{\rho}{v_i} \sum_k F_{ik} (E_{ki}(\lambda_k - \lambda_i) + \rho G_{ki} ).
\end{equation*}
Substituting $\gamma_i$ back into \eqref{eqn:eqn_F_per} and multiplying both sides by $v_i$ yields
\begin{equation*}
\begin{disarray}{l}
F_{ij}v_i(\lambda_i-\lambda_j)^2  + \rho^2 \sum_k F_{ik} (v_i G_{kj} - v_j G_{ki}) \\
+ \rho \sum_k F_{ik}\left(v_i E_{kj}(\lambda_j+\lambda_k-2\lambda_i)
- v_j E_{ki}(\lambda_k - \lambda_i) \right)  = 0.
\end{disarray}
\end{equation*}
Denoting
\begin{eqnarray}
\label{eq:sijk}
s_{jk}^i &=&\frac{1}{(\lambda_i-\lambda_j)^2}\left(E_{kj}(\lambda_j+\lambda_k-2\lambda_i)
- \frac{v_j}{v_i} E_{ki}(\lambda_k - \lambda_i) \right)\nonumber \\
t_{jk}^i &=&\frac{1}{(\lambda_i-\lambda_j)^2}\left(G_{kj} - \frac{v_j}{v_i} G_{ki}\right),
\end{eqnarray}
for $i \ne j$, and $s^i_{ik} = t^i_{ik} = 0$,
we arrive at the following perturbed linear system
\begin{eqnarray}
\label{eq:F_perturbed}
F_{ij}+\rho\sum_kF_{ik}(s_{jk}^i+\rho t_{jk}^i) &=&0, \,\,\,i\neq j\\ \nonumber
\sum_k F_{ik}v_k&=&v_i,
\end{eqnarray}
where the second set of equations $\Fb \vb = \vb$ comes, as before, from the pseudo-stochasticity constraint
$\Qb \bb{1} = \bb{1}$.
Also note that we absorbed the second-order perturbation into the terms $t^i_{jk}$.

From this point, it remains to show that the solution $\Fb$ of the perturbed system \eqref{eq:F_perturbed}
is unique and sufficiently close to the solution $\Fb_0 = \Ib$ of the unperturbed system, for which we rely on
a result in perturbation analysis of linear systems summarized as Lemma~\ref{lemma:perturbation} above.
Denoting by $\fb = (F_{11},\dots,F_{1n},\dots,F_{n1},\dots,F_{nn})^\Tr$
the row stack vector representation of $\Fb$, equation \eqref{eq:F_perturbed} can be rewritten as
\begin{equation}
(\Mb+\rho\Nb)\fb=\cb \label{eq:perturbed-system}
\end{equation}
with $\Mb=\mathrm{diag}\{\Mb_1,\dots,\Mb_n\}$ being an $n^2 \times n^2$ block-diagonal matrix, where
each $\Mb_i$ is an $n\times n$ block consisting of the identity matrix with the $i$-th row replaced by the row vector $\vb = (v_1,\dots,v_n)$.
Similarly, $\Nb$ is an $n^2 \times n^2$ block-diagonal matrix with the $n \times n$
blocks $\Nb_i = (s^i_{jk}+ \rho t^i_{jk})_{jk}$,
and $\cb$ is an $n^2\times 1$ vector of zeros, with every $(i-1)(n+1)+1$-st element replaced by $v_i$.
Due to the block-diagonal structure of $\Mb$, we readily have that $\Mb^{-1}$ is also block-diagonal
with the same structure, where each $n\times n$ diagonal block $\Mb^{-1}_i$ is
the identity matrix with the $i$-th row replaced by the row vector 
$\bb{w} = \displaystyle{\frac{1}{v_i}} (-v_1,\dots, -v_{i-1}, 1,  -v_{i+1},\dots,-v_n)$.
%
Decomposing each $\Mb^{-1}_i$ into the sum of the identity matrix and a rank-one matrix, we have
$$
\| \Mb^{-1}_i \| \le \| \bb{I} \| + \| \bb{w} \| < 1 + \frac{\sqrt{n}}{\epsilon^2},
$$
where the second inequality is due to the $(\epsilon,\delta)$-friendliness assumption.
Due to the block-diagonal structure of $\Mb^{-1}$,
\begin{equation}
\label{eq:M-BOUND}
\|\Mb^{-1}\| \le \max_{i=1,\dots,n} \|\Mb_i^{-1}\| < 1 + \frac{\sqrt{n}}{\epsilon^2}.
\end{equation}

Similarly, we obtain
\begin{eqnarray}
\|\Nb\|^2 &\le & \max_{i=1,\dots,n} \|\Nb_i \|^2_\mathrm{F}=\sum_{jk}(s_{jk}^i+\rho t_{jk}^i)^2 \nonumber\\
&\le & 2 \left( \sum_{jk}(s_{jk}^i)^2+\rho^2 (t_{jk}^i)^2 \right).
\label{eq:n-bound}
\end{eqnarray}
To bound the $(s_{jk}^i)^2$ terms, we invoke strong friendliness again, obtaining $\frac{v_i}{v_j} \le  \frac{1}{\epsilon^2}$.
Combining this result with
$(\lambda_i - \lambda_j)^2 \ge \delta^2$ for $i \ne j$, $\lambda_i^2 \le \specrad^2$, and
substituting into \eqref{eq:sijk} yields
\begin{eqnarray}
\label{eq:sijk2}
(s_{jk}^i)^2 &\leq& \left( \frac{2\specrad}{\delta^2}\left(2E_{kj} +  \frac{1}{\epsilon^2}E_{ki} \right) \right)^2  \\
& \le &  \frac{4\specrad^2}{\delta^4}\left(2E_{kj}^2 + \frac{4}{\epsilon^2}|E_{kj}E_{ki}| + \frac{1}{\epsilon^4} E^2_{ki} \right). \nonumber
\end{eqnarray}
For the first term, we use

\begin{eqnarray*}
\sum_{jk} E_{kj}^2 = \| \Eb \|^2_\mathrm{F} = \| \Ub^\Tr \Rb \Ub \|^2_\mathrm{F} = \| \Rb \|^2_\mathrm{F} \le 1,
\end{eqnarray*}
from where $ E^2_{ki} \le 1$. Using standard norm inequalities,
$$
\begin{disarray}{ll}
\sum_{jk} |E_{kj} E_{ki}| \le \sum_{jk} |E_{kj}| \le n.
\end{disarray}
$$
Substituting the latter bounds into \eqref{eq:sijk2} yields
\begin{eqnarray}
\label{eq:sum_sijk2}
\sum_{ijk} (s_{jk}^i)^2 &\leq& \frac{4\specrad^2}{\delta^4} \left( 2 + \frac{4n}{\epsilon^2} + \frac{1}{\epsilon^4} \right).
\end{eqnarray}

The $(t_{jk}^i)^2$ terms in \eqref{eq:n-bound} are bounded in a similar way. First, we observe that
$$
(t_{jk}^i)^2\leq \frac{2}{\delta^4}\left(G_{kj}^2+\frac{1}{\epsilon^4} G_{ki}^2\right).
$$
Substituting $\Gb = \Eb^2$ leads to
$$
\sum_{jk}G_{kj}^2= \|\Gb\|^2_{\mathrm{F}} \le \|\Eb\|^4_{\mathrm{F}} \le 1,
$$
from where
\begin{equation}
\label{eq:sum_tijk2}
\sum_{jk}(t_{jk}^i)^2\leq \frac{2}{\delta^4}\left(1+\frac{1}{\epsilon^4}\right).
\end{equation}
Substituting \eqref{eq:sum_sijk2} and \eqref{eq:sum_tijk2} into \eqref{eq:n-bound} and assuming $\rho^2 \le 2\specrad^2$ yields
\begin{eqnarray}
\label{eq:N-BOUND}
\|\Nb\|^2 &\leq&\frac{4}{\delta^4}\left(2\specrad^2\left(2+\frac{4n}{\epsilon^2}+\frac{1}{\epsilon^4}\right)+\rho^2\left(1+\frac{1}{\epsilon^4}\right)\right) \nonumber\\
&\leq&\frac{8 \specrad^2}{\delta^4}\left(3+\frac{4n}{\epsilon^2}+\frac{2}{\epsilon^4}\right).
\end{eqnarray}

Combining bounds \eqref{eq:M-BOUND} and \eqref{eq:N-BOUND} and requiring $\epsilon \le 1$, one has
\begin{eqnarray}
\left\|\Mb^{-1}\right\| \|\Nb\| &<& \frac{\sqrt{8} \specrad (\epsilon^2 + \sqrt{n})}{\delta^2 \epsilon^4} \sqrt{ 3\epsilon^4+4\epsilon^2 n +2 }  \nonumber\\
&\le &  \frac{\sqrt{13}}{2}(1+\sqrt{2}) \frac{\specrad n}{\delta^2 \epsilon^4}.
\end{eqnarray}
It is easy to verify that for $n \ge 2$, demanding $\displaystyle{\rho < \frac{\delta^2 \epsilon^4}{12  \specrad n^{1.5}}}$ implies
\begin{eqnarray}
\rho < 
 \frac{1}{(1+2\sqrt{n})\left\|\Mb^{-1}\right\| \|\Nb\| },
 \label{eq:ratio_}
\end{eqnarray}
from where it follows that
\begin{eqnarray}
\frac{\rho \left\|\Mb^{-1}\right\| \|\Nb\| }{1-\rho \left\|\Mb^{-1}\right\| \|\Nb\| } < \frac{1}{2 \sqrt{n}}. \label{eq:ratio}
\end{eqnarray}
Since \eqref{eq:ratio_} implies $\| \rho \Mb^{-1} \Nb\|  < \rho \left\|\Mb^{-1}\right\| \|\Nb\| < 1$, $\Ib + \rho \Mb^{-1} \Nb$ is invertible, and so is $\Mb(\Ib + \rho \Mb^{-1} \Nb) = \Mb + \rho \Nb$, from which uniqueness of the solution follows.
Invertibility of the perturbed system \eqref{eq:perturbed-system} allows to invoke Lemma~\ref{lemma:perturbation}, which combined with \eqref{eq:ratio} yields
\begin{equation}
\|\Fb-\Fb_0\|_\mathrm{F}  = \|\fb - \fb_0\| \leq
\frac{\rho \left\|\Mb^{-1}\right\| \|\Nb\| }{1-\rho \left\|\Mb^{-1}\right\| \|\Nb\| } \|\fb_0\| < \frac{1}{2},
\end{equation}
where we used $\| \fb_0 \| = \| \Fb_0 \|_\mathrm{F} = \sqrt{n}$ since $\Fb_0 = \Ib$.
Recalling that $\Pb = \Ub\Fb\Ub^\Tr\Pib^\ast$ and that the unperturbed solution is $\Pb_0=\Pib^\ast$,
and using the orthonormality of $\Ub$ and $\Pib^\ast$,
one has $\| \Pb  - \Pib^\ast \|_\mathrm{F}  = \|\Ub(\Fb-\Fb_0)\Ub^\Tr\Pib^\ast\|_\mathrm{F} = \|\Fb-\Fb_0\|_\mathrm{F}$, which completes the proof.
\end{proof}

\begin{proof}[Proof of Theorem~\ref{theorem:symmetric}]
The proof goes along the lines of the proof of Theorem~\ref{thm:asym:iso}. We reparametrize the optimization problem \eqref{eq:gm_relaxed_seeded} in terms
of $\Qb=\Pb {\Pib^{\ast\Tr}}$ instead of $\Pb$. The assumption that $\Pib^{\ast\Tr}\bb{C} = \bb{D}$ allows to rewrite the second term of the objective as $\mu \| \Pb \Cb - \Db \|_\mathrm{F}^2 = \mu \| \Qb \Db - \Db \|_\mathrm{F}^2$,
yielding the following first-order optimality condition:
\begin{equation}
\Qb\tilde{\Bb}^2+\Bb^2\Qb-2\Bb\Qb\tilde{\Bb}+ \mu (\bb{Q}-\bb{I}) \bb{D}\bb{D}^\Tr \alphab\ones^\Tr  = 0.
\end{equation}
Denoting by $\Bb = \Ub \Lambdab \Ub^\Tr$ the orthonormal eigendecomposition of $\Bb$, and multiplying by $\Ub^\Tr$ from the left and by $\Ub$ from the right yields
\begin{equation}
\Fb\Lambdab^2+\Lambdab^2\Fb-2\Lambdab\Fb\Lambdab + \mu \Fb \Gb - \mu \Gb + \gammab \vb^\Tr=0,
\label{eq:Fsym}
\end{equation}
with $\Gb = \Ub^\Tr \Db \Db^\Tr \Ub$, to which, as before, we add the pseudo-stochasticity constraint $\Fb \vb = \vb$.
System \eqref{eq:Fsym} can be expressed coordinate-wise as
\begin{equation}
\label{eqn:eqn_Fsym_}
F_{ij}(\lambda_i-\lambda_j)^2 + \mu \sum_{k} F_{ik} G_{kj} - \mu G_{ij} +v_j\gamma_i=0.
\end{equation}
Note that both the system and the constraint decouple into $n$ independent systems whose variables are the rows $\bb{f}_i = (F_{i1},\dots,F_{in})^\Tr$ of $F$. 

Let us fix $i$ and distinguish between two cases: First, if $v_i \ne 0$ ($\bb{u}_i$ does not belong to a hostile eigenspace), setting $j = i$ yields
\begin{equation}
\gamma_i = \frac{\mu}{v_i} \left( G_{ii} - \sum_k F_{ik} G_{ki} \right).
\end{equation}
Substituting this result into  \eqref{eqn:eqn_Fsym_}, we obtain
\begin{eqnarray}
\label{eqn:eqn_Fsym__}
\lefteqn{F_{ij}(\lambda_i-\lambda_j)^2 + \mu \sum_{k} F_{ik} G_{kj} - \mu G_{ij}} \\
&&  + \mu \frac{v_j}{v_i} \left( G_{ii} - \sum_k F_{ik} G_{ki} \right) =0. \nonumber
\end{eqnarray}
This can be further rewritten as the $n \times n$ system $\bb{M}_i \bb{f}_i = \bb{c}_i$, where
\begin{eqnarray}
\lefteqn{\bb{M}_i = \mathrm{diag}\{  (\lambda_i - \lambda_1)^2, \dots, (\lambda_i - \lambda_1)^2  \}} \\
&& + \mu \left( \Ib - \frac{1}{v_i} \vb \eb_i^\Tr \right) \Gb + \eb_i \vb^\Tr \nonumber \\
\lefteqn{\bb{c}_i = \mu \Gb \eb_i - \mu \frac{G_{ii}}{v_i} \vb + \eb_i \eb_i^\Tr \vb}
\end{eqnarray}
and $\eb_i$ denotes the $i$-th standard Euclidean basis vector.
Note that since \eqref{eqn:eqn_Fsym__} gives a trivial equation for for $j=i$, we replaced it by the pseudo-stochasticity constraint $\vb^\Tr \fb_i = v_i$, expressed by the last terms of $\bb{M}_i$ and $\bb{c}_i$ above. 

The assumption that $\bb{C}$ and $\bb{D}$ are covariant under the isomorphism $\Pib^\ast$ makes the above system consistent in the sense that $\fb_i = \eb_i$ is its solution; it remains to show that the latter is the only solution, that is, $\Mb_i$ is full rank.
Denoting the matrix $\Mb_i$ with $\mu = 0$ by $\Mb_i^0$, we observe that it is full rank if and only if $\lambda_i$ is a simple eigenvalue; otherwise, if it has multiplicity $m_i>1$, $\Mb_i^0$ is rank-$m_i$ deficient with the null space  $\mathrm{null}(\Mb_i^0) =  \mathrm{span}\{ \eb_{i+1}, \dots, \eb_{i+m_i} \}$ corresponding to the vanishing rows of $\Mb_i^0$.  In order to make $\Mb_i$ full rank for $\mu>0$, the range of $\bb{R} = \left( \Ib - \frac{1}{v_i} \vb \eb_i^\Tr \right) \Gb$ has to contain the latter null space, which happens if for every $j = i+1,\dots, i+m_i$, $\bb{R} \eb_j \ne \bb{0}$. Rearranging terms in $\bb{R}$ yields $v_i \Gb \eb_j \ne \vb \eb_i^\Tr \Gb \eb_j$; substituting the defintion of $\Gb$ in terms of $\Db$ and $\Ub$, and using $\ub_i = \Ub \eb_i$ and $v_i = \ones^\Tr \ub_i$ yields
\begin{eqnarray}
(\ones^\Tr \ub_i) \Ub^\Tr \Db \Db^\Tr \ub_j &\ne& \vb \eb_i^\Tr \Ub^\Tr \Db \Db^\Tr \ub_j,
\end{eqnarray}
from where the first condition of the theorem follows.

In the second case where $v_i = 0$ ($\bb{u}_i$ belongs to a hostile eigenspace), the Lagrange multiplier $\gamma_i$ remains undetermined and the system $\bb{M}_i \bb{f}_i = \bb{c}_i$ is defined by
\begin{eqnarray}
\bb{M}_i &=& \mathrm{diag}\{  (\lambda_i - \lambda_1)^2, \dots, (\lambda_i - \lambda_1)^2  \} + \mu \Gb + \eb_i \vb^\Tr \nonumber\\
\bb{c}_i &=& \mu \Gb \eb_i - \gamma_i \vb. 
\end{eqnarray}
Now, if $\lambda_i$ has multiplicity $m_i$, $\Mb_i^0$ is rank-$(m_i+1)$ deficient with the null space  $\mathrm{null}(\Mb_i^0) =  \mathrm{span}\{ \eb_{i}, \dots, \eb_{i+m_i} \}$, as the $i$-th row of $\Mb_i^0$ also vanishes.
For $\mu>0$, the system becomes full rank if the range of $\bb{G}$ contains the latter null space, which yields the second condition of the theorem.

\end{proof}






\bibliographystyle{alpha}
\bibliography{gm}

\end{document}